\documentclass[12pt,reqno]{amsart}
\usepackage[left=3.2cm, top=3.4cm, bottom=3.4cm, right=3.2cm]{geometry}
\usepackage{mathrsfs}
\usepackage{amsmath,amsxtra,amsfonts,amssymb, amsthm, amscd, epsfig}
\usepackage[dvipsnames,usenames]{color}
\usepackage{amssymb,amsmath,amsfonts,epsfig}
\usepackage{graphicx}

\numberwithin{equation}{section}

\theoremstyle{plain}
\newtheorem{thm}{Theorem}[section]

\newtheorem{prop}[thm]{Proposition}
\newtheorem{cor}[thm]{Corollary}

\newtheorem{lem}[thm]{Lemma}
\newtheorem{de}[thm]{Definition}
\newtheorem{rem}[thm]{Remark}
\newtheorem{ex}[thm]{Example}
\newcommand{\eqa}{\begin{eqnarray}}
\newcommand{\eeqa}{\end{eqnarray}}
\newcommand{\beq}{\begin{equation}}
\newcommand{\eeq}{\end{equation}}
\newcommand{\nn}{\nonumber}
\newcommand{\p}{\partial}

\newcommand{\be}{{\bf 1_\mathfrak{F}}}
\def \la {\langle}
\def \ra{\rangle}
\def \var{\varepsilon}
\def \trm{\mathrm{tr}_{2,k}^\var \,}

\def \tr{\mathrm{tr}_\mathfrak{F} \,}

\def \dsum{\displaystyle\sum}

\begin{document}
\title[]
{The Frobenius-Virasoro algebra and Euler equations }
\author[]{Dafeng Zuo}

\address[]{School of Mathematical Science,University of Science and
Technology of China, Hefei 230026, P.R.China}

\address[]{Wu Wen-Tsun Key Laboratory of Mathematics,
USTC, Chinese Academy of Sciences}

\address[]{School of Mathematics and Statistics, University of Glasgow}

\email{dfzuo@ustc.edu.cn}


\date{\today}

\begin{abstract}
We introduce an $\mathfrak{F}$-valued generalization of the Virasoro algebra,
called the Frobenius-Virasoro algebra $\mathfrak{vir_F}$, where  $\mathfrak{F}$ is
a Frobenius algebra over $\mathbb{R}$. We also study Euler equations on the regular dual
of $\mathfrak{vir_F}$, including the $\mathfrak{F}$-$\mathrm{KdV}$ equation and the
$\mathfrak{F}$-$\mathrm{CH}$ equation and the $\mathfrak{F}$-$\mathrm{HS}$ equation,
and discuss their Hamiltonian properties.
 \end{abstract}

\maketitle 
\tableofcontents
\section{Introduction}

Let $\mathfrak{G}$ be a Lie algebra and $\mathfrak{G}^*$
(the regular part of) its dual, and let $\la~,~\ra^*$ denote a natural pairing between
$\mathfrak{G}$ and $\mathfrak{G}^*$.

\begin{de} The Euler equation on $\mathfrak{G}^*$ is defined by
the following system (e.g., \cite{AK1998,KM2003}):
\beq \frac{dm}{dt}=-ad^*_{\mathcal{A}^{-1}m}m,\label{eq0.1}\eeq
as an evolution of a point $m \in \mathfrak{G}^*$, where
$\mathcal{A}:\mathfrak{G}\rightarrow \mathfrak{G}^*$ is an invertible self-adjoint operator,
called the inertia operator.\end{de}
It is well known that the KdV equation
\beq u_t+3uu_x+c u_{xxx}=0\nn\eeq
and the Camassa-Holm (CH in brief) equation
\beq m_t+2mu_x+m_xu+c u_{xxx}=0,\quad m=u-u_{xx} \nn\eeq
and the Hunter-Saxton (HS in brief)  equation
\beq m_t+2mu_x+m_xu+c u_{xxx}=0,\quad m=-u_{xx} \nn \eeq
could be regarded as Euler equations on the dual of Virasoro
algebra $\mathfrak{vir}$ with different inner products (\cite{KM2003,Kolev2007,M1998,OK1987}).
Let us remark that V.I.Arnold in \cite{Arn1966} suggested a general framework
for the Euler equation on an arbitrary Lie group $G$, which
is useful to characterize a variety of conservative dynamical systems,
please see $e.g.$, \cite{ AK1998, Guha2006,KM2003,KW2009,Kolev2007,M1998,OK1987,IS2013,
Zuo-2010-1,Zuo-2013}
and references therein. If the  corresponding Lie algebra is $\mathfrak{G}$, then the Euler
equation \eqref{eq0.1} on $\mathfrak{G}^*$ could describe a geodesic flow w.r.t a suitable
one-side invariant Riemannian metric on  Lie group $G$.

In our recent works \cite{SZ2014,Zuo2013}, we studied the relation between Frobenius manifolds
and Frobenius algebra-valued integrable systems.
\begin{de} A Frobenius algebra $(\mathfrak{F},g_\mathfrak{F},{\bf 1_\mathfrak{F}},\circ)$ over $\mathbb{R}$ is a
 free $\mathbb{R}$-module $\mathfrak{F}$
of finite rank $l$, equipped with a commutative and associative
multiplication $\circ$  and a unit ${\bf 1_\mathfrak{F}}$, and a $\mathbb{R}$-bilinear symmetric
nondegenerate form $g_\mathfrak{F}:\mathfrak{F}\times \mathfrak{F}\rightarrow \mathbb{R}$
satisfying $g_\mathfrak{F}(a\circ b,c)=g_\mathfrak{F}(a,b\circ c)$.\end{de}
Having this nondegenerate form $g_\mathfrak{F}$ is equivalent to having a linear form
$\tr:\mathfrak{F}\rightarrow \mathbb{R}$ whose kernel contains no
trivial ideas. This linear form is often called a trace map.
Indeed, given $g_\mathfrak{F}$, we put $\tr(a):=g_\mathfrak{F}(a,{\bf 1_\mathfrak{F}})$. Conversely, given
$\tr$, we could define  $g_\mathfrak{F}(a,b):=\tr(a\circ b)$.

Observe that an $\mathfrak{F}$-valued KdV ($\mathfrak{F}$-KdV) equation
\beq u_t+3u\circ u_x+\zeta \circ u_{xxx}=0, \quad \zeta\in\mathfrak{F} \label{eq0.5}\eeq
has been derived in \cite{SZ2014,Zuo2013}, where $u$ is a smooth $\mathfrak{F}$-valued function.
A natural question is to ask:

{\it``Could the $\mathfrak{F}$-KdV equation \eqref{eq0.5} be regarded as
a Euler equation on the regular dual of an infinite-dimensional Lie algebra $\mathfrak{G}$ ?" }

Our work is inspired by this question. This paper
is to give an affirmative answer and organized as follows.
Firstly, we introduce an infinite dimensional Lie algebra,
called {\it the Frobenius-Virasoro algebra} $\mathfrak{vir_F}$,
which is an $\mathfrak{F}$-valued generalization of the Virasoro algebra. Afterwards,
we compute Euler equations on the regular dual $\mathfrak{vir_F}^*$ of $\mathfrak{vir_F}$
for certain products, including the $\mathfrak{F}$-KdV equation, the $\mathfrak{F}$-CH equation
and the $\mathfrak{F}$-HS equation. Moreover we show that {\it all resulted Euler equations
for the inner product $P_{\alpha,\beta}$ are local bihamiltonian}.
Let us remark that in order to define the Euler equation on $\mathfrak{vir_F}^*$, it is enough to
require a commutative and associative algebra  $(\mathfrak{F},{\bf 1_\mathfrak{F}},\circ)$.
In other words, we don't require the existence of trace map $\tr$. An interesting fact
(also noted in \cite{SZ2014}) is that
{\it if on $(\mathfrak{F},{\bf 1_\mathfrak{F}},\circ)$, there are many different trace maps,
then the corresponding $\mathfrak{F}$-valued  Euler equation has many different (bi)hamiltonian structures.}
Finally we discuss some examples to illustrate our construction.

\section{Euler equations on $\mathfrak{vir_F}^*$}
Throughout this paper, we assume that the Frobenius algebra $\mathfrak{F}
:=(\mathfrak{F},\tr,{\bf 1_\mathfrak{F}},\circ)$ has the basis
$e_1=\be,\, e_2,\cdots, e_l$.

\subsection{The Frobenius-Virasoro algebra $\mathfrak{vir_F}$} We begin with some definitions.
\begin{de}We define an infinite-dimensional Lie algebra $(\mathfrak{X},[~,~])$ over $\mathbb{R}$ by
\beq \mathfrak{X}:=\left\{u(x)\frac{d}{dx}| u\in \mathrm{C}^\infty(\mathbb{S}^1,\mathfrak{F})\right\},
\quad [u\p,v\p]:=(u{\circ} v_x-u_x{\circ} v)\p, \quad \p=\frac{d}{dx}.\eeq
 \end{de}

We remark that $\mathfrak{X}$ is different from
the loop algebra $L\mathfrak{F}$ of $\mathfrak{F}$.
As vector spaces, they are isomorphic under the map
$$\Psi: L\mathfrak{F}\rightarrow \mathfrak{X},\quad \Psi(u)=u\p.$$
 But as Lie algebras, $\Psi$ is not a Lie algebra homomorphism.

\begin{lem}The map $\omega_\mathfrak{F}: \mathfrak{X}\times \mathfrak{X} \rightarrow \mathfrak{F}$ defined by
\beq \omega_\mathfrak{F}(u\p, v\p)= \int_{\mathbb{S}^1} u{\circ} v_{xxx}dx \eeq
is a nontrivial 2-cocycle on $\mathfrak{X}$, called the $\mathfrak{F}$-valued Gelfand-Fuchs cocycle.
 \end{lem}
 \begin{proof}Observe that the Frobenius algebra $\mathfrak{F}$ is commutative and associative,
then we have
$$\mbox{(i). $\omega_\mathfrak{F}(u\p, v\p)=-\omega_\mathfrak{F}(v\p, u\p)$;\quad
(ii). $\omega_\mathfrak{F}(u\p, [v\p,w\p])+c.p.=0$},$$
which follow the desired result. \end{proof}

\begin{de} \label{de2.4} The central extension of $\mathfrak{X}$ is
called the Frobenius-Virasoro algebra, denoted by $\mathfrak{vir_F}$ with the Lie bracket
\beq [(u\p,a),(v\p,b)]:=\left([u\p,v\p], \,
\omega_\mathfrak{F}(u\p,v\p)\right). \eeq
 \end{de}

 When one chooses the Frobenius algebra $\mathfrak{F}$ to be
$\mathbb{R}$, $\mathfrak{vir_F}$ is exactly the Virasoro algebra and
$\mathfrak{X}=\mathrm{Vect}(\mathbb{S}^1)$. It is well known
(e.g.\cite{KW2009,PS1986}) that the second continuous cohomology group
$\mathrm{H}^2(\mathrm{Vect}(\mathbb{S}^1),\mathbb{R})\cong \mathbb{R}$
is generated by the Gefland-Fuchs cocycle.
Generally when $\rm{dim}\mathfrak{F}>1$,  $\mathrm{H}^2(\mathfrak{X},\mathfrak{F})$
is not generated by the $\mathfrak{F}$-valued Gelfand-Fuchs cocycle $\omega_{\mathfrak{F}}$.
An interesting problem is to compute $\mathrm{H}^2(\mathfrak{X},\mathfrak{F})$.

\subsection{Euler equations on $\mathfrak{vir_F}^*$}
We denote the regular dual of the Frobenius-Virasoro algebra
 $\mathfrak{vir_F}$ by
\beq \mathfrak{vir_F}^*=\left\{(m(x,t)(dx)^2, \zeta(t))|\mbox {$m(x,t)$ and $\zeta(t)$ are smooth $\mathfrak{F}$-valued functions } \right\}\nn\eeq
with respect to the paring
 \beq \la (m dx^2, \zeta), (u\p, a) \ra^*=\tr\int_{\mathrm{S}^1}m{\circ} u dx+\tr(\zeta{\circ} a).\eeq
Write  $\hat{m}=(m dx^2, \zeta) \in \mathfrak{vir_F}^*$ and $\hat{u}=(u\p, a),\, \hat{v}=(v\p, b)
\in \mathfrak{vir_F}$.  By the definition,
\beq  \la ad^*_{\hat{u}}(\hat{m}),\hat{v} \ra^*=-\la \hat{m},[\hat{u},\hat{v}]\ra^*
 = \tr\int_{\mathbb{S}^1}(2m{\circ}u_x+m_x{\circ}u+\zeta {\circ}u_{xxx}){\circ}v dx \nn
 \eeq
 which yields that the coadjoint action of $\mathfrak{vir_F}$ on $\mathfrak{vir_F}^*$ is given by
 \beq ad^*_{\hat{u}} \hat{m}=\left((2m{\circ}u_x+m_x{\circ}u+\zeta{\circ}u_{xxx})(dx)^2,\,0\right).\eeq
On $\mathfrak{vir_F}$, we introduce a two-parameter family of inner product $P_{\alpha,\beta}$,
$\alpha,\,\beta\in\mathfrak{F}$ defined by
\beq \la\hat{u},\hat{v}\ra= \tr\int_{\mathrm{S}^1}(\alpha{\circ}u{\circ}v+\beta{\circ}u_x{\circ}v_x) dx+\tr(a{\circ}~b).
\eeq
Observe that for the $P_{\alpha,\beta}$, the inertia operator
 $\mathcal{A}:\mathfrak{vir_F} \longrightarrow \mathfrak{vir_F}^*$
 is defined by
$\la \hat{u},\hat{v} \ra=\la \mathcal{A}(\hat{u}),\hat{v} \ra^*$. In other words,
$\mathcal{A}(\hat{u})=(\Lambda(u),a)$,
where $\Lambda=\alpha-\beta\p^2$ is an $\mathfrak{F}$-valued
differential operator. So we have

\begin{prop}The Euler equation \eqref{eq0.1} on $\mathfrak{vir_F}^*$ for
$P_{\alpha,\beta}$ reads
\beq m_t+2m{\circ}u_x+m_x{\circ}u+\zeta{\circ}u_{xxx}=0,\quad \zeta_t=0, \label{eq3.5} \eeq
where $m=\Lambda(u)=\alpha {\circ}u-\beta {\circ}u_{xx}$.
 \end{prop}

When $\zeta=0$,  the system \eqref{eq3.5} could be regarded as the Euler equation
on $\mathfrak{X}^*$.  When the Frobenius algebra $\mathfrak{F}$ is one-dimensional, i.e.,
$\mathbb{R}$, the system \eqref{eq3.5} is the Euler equation on $\mathfrak{vir}^*$ (e.g.,
\cite{KM2003}). Generally, when $\alpha\ne 0$, $\beta=0$ and $0\ne \zeta \in \mathfrak{F}$,
the system \eqref{eq3.5} reads \mbox{the $\mathfrak{F}$-KdV equation}
\beq \alpha{\circ}u_t+3\alpha{\circ}u{\circ}u_x+\zeta{\circ}u_{xxx}=0.\label{F1} \eeq
When $\alpha\ne 0$, $\beta\ne 0$ and  $\zeta \in \mathfrak{F}$, the system \eqref{eq3.5} becomes
\mbox{the $\mathfrak{F}$-CH equation}
\beq \quad m_t+2m{\circ}u_x+m_x{\circ}u+\zeta{\circ}u_{xxx}=0,\quad m=\alpha {\circ}u-\beta{\circ}u_{xx}
.\label{F2}\eeq
When $\alpha=0$, $\beta\ne 0$ and $ \zeta \in \mathfrak{F}$, the system \eqref{eq3.5}
reduces to \mbox{the $\mathfrak{F}$-HS equation}
 \beq  \beta{\circ}(u_{xxt}+2u_{xx}{\circ}u_x+u_{xxx}{\circ}u)-\zeta {\circ}u_{xxx}=0,
 \quad m=-\beta{\circ}u_{xx}.\label{F3}\eeq

\begin{ex}\label{ex2.2} Let $\mathcal{Z}_2^{\var}$ be a $2$-dimensional commutative and associative algebra
over $\mathbb{R}$ with the basis $e_1=\be, e_2$ satisfying
\beq e_1{\circ} e_1=e_1,\quad e_1{\circ} e_2=e_2, \quad e_2{\circ} e_2=\var e_1, \quad \var \in \mathbb{R}. \nn \eeq
Thus for any $A\in \mathcal{Z}^\var_2$, we could write $A=a_1e_1+a_2e_2$, $a_k \in \mathbb{R}$
and define two ``basic" trace-type maps as follows
\beq \trm(A)=a_k+a_2(1-\delta_{k,2})\delta_{\var,0} \quad k=1,2.\label{ZZ2.3} \eeq
So $(\mathcal{Z}_2^{\var}, \trm, \be,{\circ})$ for $k=1,2$ are the Frobenius algebras (\cite{SZ2014}).
The $\mathcal{Z}_2^{\var}$-valued Euler equation with $\zeta\in \mathcal{Z}_2^{\var}$ is given by
 \beq m_t+2m\circ u_x+m_x\circ u+\zeta \circ u_{xxx}=0,\quad m=\alpha\circ u-\beta\circ u_{xx}. \label{Z3.15}\eeq

\noindent{\bf (i).}\, When $\alpha=\zeta=\be$ and $\beta=0$, the system
\eqref{Z3.15} reduces to the $\mathcal{Z}_2^{\var}$-KdV equation (\cite{SZ2014,Zuo2013})
\beq  u_t+3u\circ u_x+u_{xxx}=0,\quad u=ve_1+we_2\nn  \eeq
equivalently in componentwise forms,
\beq v_t+3vv_x+v_{xxx}+3 \var w w_x=0,\quad w_t+3(vw)_x+w_{xxx}=0.\label{zeq3.6}\eeq
When $\var=0$, the system \eqref{zeq3.6} is the coupled KdV equation in \cite{CO2006,CO2006-2,FR1989}.
When $\var=-1$, the system \eqref{zeq3.6} is a complexification of the KdV equation.\\

\noindent{\bf (ii).}\, When $\alpha=\beta=\be$ and $\zeta=0$, the system
\eqref{Z3.15} reduces to the $\mathcal{Z}_2^{\var}$-CH equation
\beq  m_t+2m\circ u_x+m_x\circ u=0,\quad m=u-u_{xx}, \quad u=ve_1+we_2,\nn  \eeq
equivalently in componentwise forms,
\beq\begin{array}{ll}  p_t+2pv_x+p_xv+\var(2qw_x+q_xw)=0, &   p=v-v_{xx},\\
q_t+2qv_x+q_xv+2pw_x+p_xw=0, & q=w-w_{xx}.\end{array}\label{zeq3.7} \eeq
When $\var=-1$, the system \eqref{zeq3.7} is the complex-CH equation (e.g.,\cite{Qu2013}).\\

\noindent{\bf (iii).}\, When $\alpha=0$ and $\beta=\zeta=\be$, the system
\eqref{Z3.15} reduces to the $\mathcal{Z}_2^{\var}$-HS equation
\beq  m_t+2m\circ u_x+m_x\circ u=0,\quad m=-u_{xx}, \quad u=ve_1+we_2,\nn  \eeq
equivalently in componentwise forms,
\beq\begin{array}{ll}  p_t+2pv_x+p_xv+\var(2qw_x+q_xw)=0, &   p=-v_{xx},\\
q_t+2qv_x+q_xv+2pw_x+p_xw=0, & q=-w_{xx}.\end{array}\label{zeq3.8} \eeq
\end{ex}

\subsection{Hamiltonian structures of the Euler equation \eqref{eq3.5}}
Let us take two arbitrary smooth functionals
$$\tilde{F}_i:\mathfrak{vir_\mathfrak{F}}^* \to \mathbb{R},\quad
\tilde{F}_i(\hat{m})=\int_{\mathbb{S}^1} \tr F_i(m) dx=\int_{\mathbb{S}^1} f_i(m_1,\cdots,m_l) dx,\quad i=1,2,$$
where $m=\dsum_{k=1}^lm_ke_k$.
 The variational derivative $\dfrac{\delta \tilde{F}_i }{\delta \hat{m}}$ is defined as
\beq \dfrac{\delta \tilde{F}_i }{\delta \hat{m}}=(\dfrac{\delta F_i }{\delta m}\p,\, 0)\in
\mathfrak{vir_\mathfrak{F}},\label{VD1}\eeq
where
$\dfrac{\delta F_i }{\delta m}$ is implicitly determined by
\eqa \tilde{F}_i(m+\delta m)-\tilde{F}_i(m)
&=&\int_{\mathbb{S}^1} \tr \left(\dfrac{\delta F_i }{\delta m}\circ \delta m+o(\delta m)\right)dx\nn\\
&=& \int_{\mathbb{S}^1}  \left(\dsum_{k=1}^l\dfrac{\delta f_i }{\delta m_k}\delta m_k+o(\delta m)\right)dx
\label{VD2} \eeqa
and $\dfrac{\delta f_i }{\delta m_k}$ is the usual variational derivative.
 This formula \eqref{VD2}
is very crucial to construct the bihamiltonian representation of the Euler equation.
On $\mathfrak{vir_\mathfrak{F}}^*$, there is a canonical Lie-Poisson bracket
\eqa
\mathcal{P}_2:=\{\tilde{F}_1,\tilde{F}_2\}_2(\hat{m})= \la m, [\dfrac{\delta \tilde{F}_1 }{\delta \hat{m}},
 \dfrac{\delta \tilde{F}_2 }{\delta \hat{m}}] \ra^*
 =\tr \int_{\mathbb{S}^1} \dfrac{\delta  {F}_1 }{\delta {m}} \circ\mathcal{J}_2
 \circ \dfrac{\delta  {F}_2 }{\delta {m}}\,dx \label{BH2}\eeqa
where $\mathcal{J}_2=-(m\p+\p m+\zeta\p^3)$ and $\hat{m}=(m dx^2, \zeta) \in \mathfrak{vir_F}^*$.
Taking a fixed point $\hat{m}_0=(\dfrac{\alpha}{2} dx^2, -\beta)$, we get another compatible Poisson bracket
denoted by
\beq \mathcal{P}_1=\{\tilde{F}_1,\tilde{F}_2\}_1(\hat{m}):
=\tr \int_{\mathbb{S}^1} \dfrac{\delta  {F}_1 }{\delta {m}}\circ \mathcal{J}_1
 \circ\dfrac{\delta  {F}_2 }{\delta {m}}\,dx, \quad \mbox{i.e.,}\quad
\mathcal{P}_1=\mathcal{P}_2|_{\hat{m}=\hat{m}_0}, \label{BH1} \eeq
where   $\mathcal{J}_1:=\mathcal{J}_2|_{\hat{m}=\hat{m}_0}=\beta \p^3-\alpha\p=-\p \Lambda$.

\begin{thm}\label{thm3.3} The $\mathfrak{F}$-valued Euler equation \eqref{eq3.5} with $\zeta\in\mathfrak{F}$
is local bihamiltonian with the freezing point
$\hat{m}_0=(\dfrac{\alpha}{2} dx^2, -\beta)\in \mathfrak{vir_{\mathfrak{F}}}^*$.\end{thm}
\begin{proof} Setting
 $$H_1=\frac{1}{2}\tr \int_{\mathbb{S}^1} m\circ u dx,\quad H_2=\frac{1}{2} \tr\int_{\mathbb{S}^1}\left({\zeta}\circ u\circ u_{xx}+\alpha\circ u^3
 -\frac{1}{2}{\beta}\circ u^2\circ u_{xx}\right) dx.$$
 With the formula \eqref{VD1}, we get
 \beq \frac{\delta H_1}{\delta u}=\Lambda(u), \quad
 \frac{\delta H_2}{\delta u}=\zeta\circ u_{xx}+\frac{3}{2}\alpha\circ u^2
 -\frac{1}{2}{\beta}\circ u_x^2-\beta\circ u\circ u_{xx}.\nn \eeq
By using $m=\Lambda(u)$, then
\beq \frac{\delta H_1}{\delta m}=\Lambda^{-1}\circ\frac{\delta H_1}{\delta u}=u,\quad
\frac{\delta H_2}{\delta m}=\Lambda^{-1}\circ \frac{\delta H_2}{\delta u}.\nn\eeq
So the system \eqref{eq3.5} could be written as
\beq m_t=\mathcal{J}_1\circ \frac{\delta H_2}{\delta m}=\mathcal{J}_2\circ\frac{\delta H_1}{\delta m}.\label{Add1}\eeq
Furthermore using the formula \eqref{VD2},
in componentwise forms the Euler equation \eqref{eq3.5} has the following bihamlitonian representation
\beq m_{k\,t}=\{m_k, H_2\}_1= \{m_k, H_1\}_2,\quad k=1,\cdots l,\eeq
where two compatible Poisson brackets $\{~,~\}_i,\, i=1,2 $ are defined in \eqref{BH2} and \eqref{BH1} respectively.
 \end{proof}

\begin{rem} Let us remark that when choose $\mathfrak{F}$ as the Frobenius algebra
 $(\mathcal{Z}_l, \mathrm{tr}_l)$ (\cite{CO2006,L2007,Zuo2013}), the Frobenius-Virasoro algebra
 $\mathfrak{vir_F}$ coincides with the polynomial Virasoro algebra introduced  by
 P.Casati and G.Ortenzi in \cite{CO2006-2}. They also computed Euler equations on $\mathfrak{vir_F}^*$ and proved
 that they admitted a local bihamiltonian structure by using the trace-type map $\mathrm{tr}_l$.
Actually in \cite{SZ2014}, it has been shown that there are at least $l$ ``basic" different ways to
regard the algebra $\mathcal{Z}_l$ as the Frobenius algebra $(\mathcal{Z}_l, \omega_{k})$
for $k=0,\cdots,l-1$. We want to mention that the trace map $\mathrm{tr}_l$ is a linear combination of ``basic"
trace maps given by $\mathrm{tr}_l=\dsum_{k=0}^{l-1}\omega_{k}-(l-1)\,\omega_{l-1}$.
Using Theorem \ref{thm3.3}, we thus obtain
\begin{cor} The $\mathcal{Z}_l$-valued
Euler equation \eqref{eq3.5} has at least $l$ ``basic" local bihamiltonian structures.
 \end{cor}\end{rem}

\subsection{Examples}
According to Example \ref{ex2.2},  $(\mathcal{Z}_2^{\var}, \trm, \be,{\circ}~)$ for $k=1,2$
are the Frobenius algebras.  We thus have

\begin{cor}The $\mathcal{Z}_2^{\var}$-valued Euler equation \eqref{Z3.15}
has at least two kinds of ``basic" local bihamiltonian structures.
\end{cor}

Naturally, we know that the $\mathcal{Z}_2^{\var}$-CH equation \eqref{F2} and the $\mathcal{Z}_2^{\var}$-HS equation
\eqref{F3} have at least two kinds of  ``basic" local bihamiltonian structures. For the $\mathcal{Z}_2^{\var}$-KdV
equation \eqref{F2}, two kinds of ``basic" local bihamiltonian structures have been obtained
in \cite{SZ2014,Zuo2013} by other methods.
Based on our construction, more precisely we have

\begin{ex}We consider the case: {\bf $[ \var \ne 0]$}.

{\bf (i). } The $\mathcal{Z}_2^{\var}$-KdV equation \eqref{zeq3.6}
\beq v_t+3vv_x+v_{xxx}+3 \var w w_x=0,\quad w_t+3(vw)_x+w_{xxx}=0\nn \eeq
could be rewritten as
\beq \left(
\begin{array}{c}
v\\
w\end{array}\right)_t=-\left(
\begin{array}{cc}
0& \p\\
\p& 0
\end{array}\right)\left(
\begin{array}{c}
\frac{\delta H_{2}}{\delta v}\\
\frac{\delta H_{2}}{\delta w}\end{array}\right)=-\left(
\begin{array}{cc}
\var J_1& J_0\\
J_0 &  J_1
\end{array}\right)\left(
\begin{array}{c}
\frac{\delta H_{1}}{\delta v}\\
\frac{\delta H_{1}}{\delta w}\end{array}\right) \nn \eeq
with Hamiltonians
 $$H_1=\int_{\mathbb{S}^1}vw dx,\quad H_2=\frac{1}{2} \int_{\mathbb{S}^1}(3v^2w+\var w^3+2vw_{xx}) dx;$$
and
\beq \left(
\begin{array}{c}
v\\
w\end{array}\right)_t=-\left(
\begin{array}{cc}
\p & 0\\
0 & \dfrac{1}{\var}\p
\end{array}\right)\left(
\begin{array}{c}
\frac{\delta \widetilde{H}_{2}}{\delta v}\\
\frac{\delta \widetilde{H}_{2}}{\delta w}\end{array}\right)=-\left(
\begin{array}{cc}
J_0& J_1\\
J_1 & \frac{1}{\var} J_0
\end{array}\right)\left(
\begin{array}{c}
\frac{\delta \widetilde{H}_{1}}{\delta v}\\
\frac{\delta \widetilde{H}_{1}}{\delta w}\end{array}\right) \nn \eeq
with Hamiltonians
 $$\widetilde{H}_1=\frac{1}{2} \int_{\mathbb{S}^1}(v^2+\var w^2) dx,\quad
 \widetilde{H}_2=\frac{1}{2} \int_{\mathbb{S}^1}(v^3+vv_{xx}+3 \var v w^2+\var w w_{xx}) dx,$$
where $J_0=\p^3+v\p+\p v$ and $J_1=w\p+\p w$.\\

{\bf (ii)}. The $\mathcal{Z}_2^{\var}$-CH equation \eqref{zeq3.7}
\beq\begin{array}{ll}  p_t+2pv_x+p_xv+\var(2qw_x+q_xw)=0, &   p=v-v_{xx},\\
q_t+2qv_x+q_xv+2pw_x+p_xw=0, & q=w-w_{xx}, \nn \end{array}\eeq
could be rewritten as
\beq \left(
\begin{array}{c}
p\\
q\end{array}\right)_t=\left(
\begin{array}{cc}
0& \p^3- \p\\
\p^3-\p& 0
\end{array}\right)\left(
\begin{array}{c}
\frac{\delta H_{2}}{\delta p}\\
\frac{\delta H_{2}}{\delta q}\end{array}\right)=-\left(
\begin{array}{cc}
\var K_1& K_0\\
K_0 &  K_1
\end{array}\right)\left(
\begin{array}{c}
\frac{\delta H_{1}}{\delta p}\\
\frac{\delta H_{1}}{\delta q}\end{array}\right) \nn \eeq
with Hamiltonians
 $$H_1=\frac{1}{2}\int_{\mathbb{S}^1}(qv+pw) dx,\quad H_2=\frac{1}{4} \int_{\mathbb{S}^1}
 \left(2vw_{xx}+2wv_{xx}-2wvv_{xx}-v^2w_{xx}-\var w^2w_{xx}
 \right) dx,$$
and
\beq \left(
\begin{array}{c}
p\\
q\end{array}\right)_t=\left(
\begin{array}{cc}
\p^3-\p & 0\\
0 & \dfrac{1}{\var}(\p^3-\p)
\end{array}\right)\left(
\begin{array}{c}
\frac{\delta \widetilde{H}_{2}}{\delta p}\\
\frac{\delta \widetilde{H}_{2}}{\delta q}\end{array}\right)=-\left(
\begin{array}{cc}
K_0& K_1\\
K_1 & \dfrac{1}{\var} K_0
\end{array}\right)\left(
\begin{array}{c}
\frac{\delta \widetilde{H}_{1}}{\delta p}\\
\frac{\delta \widetilde{H}_{1}}{\delta q}\end{array}\right) \nn \eeq
with Hamiltonians
 $$\widetilde{H}_1=\frac{1}{2} \int_{\mathbb{S}^1}(pv+\var qw) dx,\quad
 \widetilde{H}_2=\frac{1}{4} \int_{\mathbb{S}^1}\left(2vv_{xx}-v^2v_{xx}+\var(ww_{xx}-w^2v_{xx}-2vww_{xx})
 \right) dx,$$
where $K_0=p\p+\p p$ and $K_1=q\p+\p q$.\\

{\bf (iii). } The $\mathcal{Z}_2^{\var}$-HS equation \eqref{zeq3.8}
\beq\begin{array}{ll}  p_t+2pv_x+p_xv+\var(2qw_x+q_xw)=0, &   p=-v_{xx},\\
q_t+2qv_x+q_xv+2pw_x+p_xw=0, & q=-w_{xx}, \end{array}\nn\eeq
could be rewritten as
\beq \left(
\begin{array}{c}
p\\
q\end{array}\right)_t=\left(
\begin{array}{cc}
0& \p^3\\
\p^3& 0
\end{array}\right)\left(
\begin{array}{c}
\frac{\delta H_{2}}{\delta p}\\
\frac{\delta H_{2}}{\delta q}\end{array}\right)=-\left(
\begin{array}{cc}
\var K_1& K_0\\
K_0 &  K_1
\end{array}\right)\left(
\begin{array}{c}
\frac{\delta H_{1}}{\delta p}\\
\frac{\delta H_{1}}{\delta q}\end{array}\right) \nn \eeq
with Hamiltonians
 $$H_1=\frac{1}{2}\int_{\mathbb{S}^1}(qv+pw) dx,\quad H_2=\frac{1}{4} \int_{\mathbb{S}^1}
 \left(2wvp+v^2q+\var w^2q \right) dx,$$
and
\beq \left(
\begin{array}{c}
p\\
q\end{array}\right)_t=\left(
\begin{array}{cc}
\p^3& 0\\
0 & \dfrac{1}{\var}\p^3
\end{array}\right)\left(
\begin{array}{c}
\frac{\delta \widetilde{H}_{2}}{\delta p}\\
\frac{\delta \widetilde{H}_{2}}{\delta q}\end{array}\right)=-\left(
\begin{array}{cc}
K_0& K_1\\
K_1 & \dfrac{1}{\var} K_0
\end{array}\right)\left(
\begin{array}{c}
\frac{\delta \widetilde{H}_{1}}{\delta p}\\
\frac{\delta \widetilde{H}_{1}}{\delta q}\end{array}\right) \nn \eeq
with Hamiltonians
 $$\widetilde{H}_1=\frac{1}{2} \int_{\mathbb{S}^1}(pv+\var qw) dx,\quad
 \widetilde{H}_2=\frac{1}{4} \int_{\mathbb{S}^1}\left(pv^2+\var p w^2+2\var vwq
 \right) dx,$$
where $K_0=p\p+\p p$ and $K_1=q\p+\p q$.
 \end{ex}

\begin{ex}We consider another case: {\bf $[ \var=0]$}.

{\bf (i). } The $\mathcal{Z}_2^{0}$-KdV equation \eqref{zeq3.6}
\beq v_t+3vv_x+v_{xxx}=0,\quad w_t+3(vw)_x+w_{xxx}=0\nn \eeq
could be rewritten as
\beq \left(
\begin{array}{c}
v\\
w\end{array}\right)_t=-\left(
\begin{array}{cc}
0& \p\\
\p& 0
\end{array}\right)\left(
\begin{array}{c}
\frac{\delta H_{2}}{\delta v}\\
\frac{\delta H_{2}}{\delta w}\end{array}\right)=-\left(
\begin{array}{cc}
0 & J_0\\
J_0 &  J_1
\end{array}\right)\left(
\begin{array}{c}
\frac{\delta H_{1}}{\delta v}\\
\frac{\delta H_{1}}{\delta w}\end{array}\right) \nn \eeq
with Hamiltonians
 $$H_1=\int_{\mathbb{S}^1}vw dx,\quad H_2=\frac{1}{2} \int_{\mathbb{S}^1}(3v^2w+2vw_{xx}) dx;$$
and
\beq \left(
\begin{array}{c}
v\\
w\end{array}\right)_t=-\left(
\begin{array}{cc}
0 & \p\\
\p & -\p
\end{array}\right)\left(
\begin{array}{c}
\frac{\delta \widetilde{H}_{2}}{\delta v}\\
\frac{\delta \widetilde{H}_{2}}{\delta w}\end{array}\right)=-\left(
\begin{array}{cc}
0& J_0\\
J_0 & J_1-J_0
\end{array}\right)\left(
\begin{array}{c}
\frac{\delta \widetilde{H}_{1}}{\delta v}\\
\frac{\delta \widetilde{H}_{1}}{\delta w}\end{array}\right) \nn \eeq
with Hamiltonians
 $$\widetilde{H}_1=\frac{1}{2} \int_{\mathbb{S}^1}(v^2+2vw) dx,\quad
 \widetilde{H}_2=\frac{1}{2} \int_{\mathbb{S}^1}(v^3+vv_{xx}+3v^2 w+ 2v w_{xx}) dx,$$
where $J_0=\p^3+v\p+\p v$ and $J_1=w\p+\p w$.\\

{\bf (ii)} The $\mathcal{Z}_2^0$-CH equation \eqref{zeq3.7}
\beq\begin{array}{ll}  p_t+2pv_x+p_xv=0, &   p=v-v_{xx},\\
q_t+2qv_x+q_xv+2pw_x+p_xw=0, & q=w-w_{xx}  \end{array}\nn \eeq
could be rewritten as
\beq \left(
\begin{array}{c}
p\\
q\end{array}\right)_t=\left(
\begin{array}{cc}
0& \p^3-\p\\
\p^3-\p& 0
\end{array}\right)\left(
\begin{array}{c}
\frac{\delta H_{2}}{\delta p}\\
\frac{\delta H_{2}}{\delta q}\end{array}\right)=-\left(
\begin{array}{cc}
0 & K_0\\
K_0 &  K_1
\end{array}\right)\left(
\begin{array}{c}
\frac{\delta H_{1}}{\delta p}\\
\frac{\delta H_{1}}{\delta q}\end{array}\right) \nn \eeq
with Hamiltonians
 $$H_1=\frac{1}{2}\int_{\mathbb{S}^1}(qv+pw) dx,\quad H_2=\frac{1}{4} \int_{\mathbb{S}^1}
 \left(2vw_{xx}+2wv_{xx}-2wvv_{xx}-v^2w_{xx}
 \right) dx,$$
and
\beq \left(
\begin{array}{c}
p\\
q\end{array}\right)_t=\left(
\begin{array}{cc}
0& \p^3-\p\\
\p^3-\p& \p-\p^3
\end{array}\right)\left(
\begin{array}{c}
\frac{\delta \widetilde{H}_{2}}{\delta p}\\
\frac{\delta \widetilde{H}_{2}}{\delta q}\end{array}\right)=-\left(
\begin{array}{cc}
0& K_0\\
K_0 & K_1-K_0
\end{array}\right)\left(
\begin{array}{c}
\frac{\delta \widetilde{H}_{1}}{\delta p}\\
\frac{\delta \widetilde{H}_{1}}{\delta q}\end{array}\right) \nn \eeq
with Hamiltonians $$\widetilde{H}_1=\dfrac{1}{2}\int_{\mathbb{S}^1}(pv+qv+pw) dx$$
and $$\widetilde{H}_2=\frac{1}{4} \int_{\mathbb{S}^1}
 \left(2vw_{xx}+2wv_{xx}-2wvv_{xx}-v^2w_{xx}+2vv_{xx}-v^2v_{xx}
 \right) dx,$$
where $K_0=p\p+\p p$ and $K_1=q\p+\p q$.\\

{\bf (iii). }The $\mathcal{Z}_2^0$-HS equation \eqref{zeq3.8}
\beq\begin{array}{ll}  p_t+2pv_x+p_xv=0, &   p=-v_{xx},\\
q_t+2qv_x+q_xv+2pw_x+p_xw=0, & q=-w_{xx}\end{array}\nn \eeq
could be rewritten as
\beq \left(
\begin{array}{c}
p\\
q\end{array}\right)_t=\left(
\begin{array}{cc}
0& \p^3\\
\p^3& 0
\end{array}\right)\left(
\begin{array}{c}
\frac{\delta H_{2}}{\delta p}\\
\frac{\delta H_{2}}{\delta q}\end{array}\right)=-\left(
\begin{array}{cc}
0 & K_0\\
K_0 &  K_1
\end{array}\right)\left(
\begin{array}{c}
\frac{\delta H_{1}}{\delta p}\\
\frac{\delta H_{1}}{\delta q}\end{array}\right) \nn \eeq
with Hamiltonians
 $$H_1=\frac{1}{2}\int_{\mathbb{S}^1}(qv+pw) dx,\quad H_2=\frac{1}{4} \int_{\mathbb{S}^1}
 \left(2wvp+v^2q \right) dx,$$
and
\beq \left(
\begin{array}{c}
p\\
q\end{array}\right)_t=\left(
\begin{array}{cc}
0& \p^3-\p\\
\p^3-\p& \p-\p^3
\end{array}\right)\left(
\begin{array}{c}
\frac{\delta \widetilde{H}_{2}}{\delta p}\\
\frac{\delta \widetilde{H}_{2}}{\delta q}\end{array}\right)=-\left(
\begin{array}{cc}
0& K_0\\
K_0 & K_1-K_0
\end{array}\right)\left(
\begin{array}{c}
\frac{\delta \widetilde{H}_{1}}{\delta p}\\
\frac{\delta \widetilde{H}_{1}}{\delta q}\end{array}\right) \nn \eeq
with Hamiltonians
 $$\widetilde{H}_1=\frac{1}{2}\int_{\mathbb{S}^1}(pv+qv+pw) dx,\quad \widetilde{H}_2=\frac{1}{4} \int_{\mathbb{S}^1}
 \left(pv^2+2wvp+v^2q \right) dx,$$
where $K_0=p\p+\p p$ and $K_1=q\p+\p q$.
 \end{ex}

 \subsection{Euler equations on $\mathfrak{vir_F}^*$ for general product $P_{\alpha_0,\cdots,\alpha_n}$}
 To end up this section, on $\mathfrak{vir_F}$ we introduce a general product $P_{\alpha_0,\cdots,\alpha_n}$
given by
\beq \la\hat{u},\hat{v}\ra= \tr\int_{\mathrm{S}^1}\Big(\alpha_0{\circ}u{\circ}v+\dsum_{k=1}^n\alpha_k{\circ}u^{(k)}{\circ}v^{(k)}\Big) dx+\tr(a{\circ}~b),
\quad u^{(k)}=\dfrac{d^k u}{dx^k}.
\eeq
By analogy with the above discussions, we have

\begin{prop}\label{prop2.12} The Euler equation \eqref{eq0.1} on $\mathfrak{vir_F}^*$  for
$P_{\alpha_0,\cdots,\alpha_n}$ reads
\beq m_t+2m{\circ}u_x+m_x{\circ}u+\zeta{\circ}u_{xxx}=0,\quad \zeta_t=0, \label{eq4.5} \eeq
where $m=\alpha_0 {\circ}u+\dsum_{k=1}^n(-1)^k\alpha_k {\circ}u^{(2k)}$. Moreover,
the system \eqref{eq4.5} with $\zeta\in\mathfrak{F}$ could be written as
\beq m_{k,t}=\{m_k,H_1\}_2,\quad H_1=\frac{1}{2}\tr \int_{\mathbb{S}^1} m\circ u dx\eeq
where $\{~,~\}_2$ is defined in \eqref{BH2}.
\end{prop}

Generally, when $n\geq 2$, the system \eqref{eq4.5} isn't a bihamiltonian system. But
if there are many different ways to realize the algebra $(\mathfrak{F}, {\bf 1_\mathfrak{F}},\circ)$
 as the Frobenius algebras, then it follows from Proposition \ref{prop2.12} that the system \eqref{eq4.5} has many different Hamiltonian structures. For instance,

\begin{cor}The $\mathcal{Z}_2^{\var}$-valued Euler equation \eqref{eq4.5}  with
$\zeta\in \mathcal{Z}_2^{\var}$
 admits at least  two ``basic" local Hamiltonian structures.
\end{cor}

\section{Conclusion}
In order to understand Eulerian nature of the $\mathfrak{F}$-valued KdV equation, we have introduced
the Frobenius-Virasoro algebra $\mathfrak{vir_\mathfrak{F}}$ and also described
 Euler equations on $\mathfrak{vir_\mathfrak{F}}^*$ under the product $P_{\alpha_0,\cdots,\alpha_n}$
 and proved that all resulted Euler equations for $P_{\alpha,\beta}$ are local bihamiltonian systems.
Here we only studied the Euler equation associated with $\mathfrak{vir_\mathfrak{F}}$.
In subsequent publications we hope to address those problems related to algebraic properties of  $\mathfrak{vir_\mathfrak{F}}$, such as  \\
\qquad {\it Q1. What is the second continuous cohomology group $\mathrm{H}^2(\mathfrak{X},\mathfrak{F})$?}\\
  {\it Q2. How about the representation theory of $\mathfrak{vir_\mathfrak{F}}$?}  \\
  {\it Q3. If exists, what is the corresponding Lie group $G_\mathfrak{F}$ of $\mathfrak{vir_\mathfrak{F}}$? For instance, $G_\mathbb{R}$ is  the Bott-Virasoro group.}

\medskip

\noindent{\bf Acknowledgements.}  The author is grateful to Professors Qing Chen, Yi Cheng
and Youjin Zhang for constant supports and Professor Ian A.B.Strachan, Dr.Ying Shi for
fruitful discussions. This work is partially supported by NCET-13-0550, NSFC (11271345,
11371138), SRF for ROCS, SEM and OATF,USTC.



\begin{thebibliography}{99}


\bibitem{Arn1966} V.I.Arnold, \emph{Sur la g$\acute{e}$om$\acute{e}$trie diff$\acute{e}$rentielle des groupes
de Lie de dimenson infinie et ses applications $\grave{a}$ l'hydrodynamique des fluids parfaits}, Ann. Inst.
Fourier Grenoble 16 (1966) 319--361.


\bibitem{AK1998} V.I.Arnold and B.Khesin, \emph{Topological methods in
hydrodynamics}, Applied Mathematical Sciences, vol. 125, Springer-Verlag, New York, (1998) pp. xv+374.

\bibitem{CO2006}P.Casati and G.Ortenzi,\emph{New integrable hierarchies from
vertex operator representations of polynomial Lie algebras.} J.Geom.Phys.
56(2006)418--449.

\bibitem{CO2006-2}P.Casati and G.Ortenzi, \emph{Bihamiltonian Equations on Polynomial Virasoro algebras.}
J.Nonl. Math. Phys., 3(2006) 352--364.


\bibitem{FR1989}A.P.Fordy, A.G.Reyman and M.A.Semenov-Tian-Shansky,
 \emph{Classical $r$-matrices and compatible Poisson brackets for coupled KdV
systems}, Lett. Math. Phys. 17 (1989) 25--29.

\bibitem{Guha2006}P.Guha and P.J.Olver, \emph{Geodesic flow and two (super) component
analog of the Camassa-Holm equation},
Symmetry Integrability Geom. Methods Appl. (SIGMA), 2 (2006) Paper 054, 9 pp.

\bibitem{KM2003}B.Khesin and G.Misiolek,
\emph{Euler equations on homogeneous spaces and Virasoro orbits}, Adv. Math.,
176(2003)116--144.

\bibitem{KW2009} B.Khesin and R.Wendt, \emph{The Geometry of Infinite-Dimensional
Groups}, Springer-Verlag, New York, 2009.

\bibitem{Kolev2007}B.Kolev, \emph{Bihamiltonian systems on the dual of the Lie
algebra of vector fields of the circle and periodic shallow water equations},
Phil. Trans. R. Soc. A 365 (2007)2333--2357.

\bibitem{L2007}J.van de Leur, \emph{B\"acklund transformations for new integrable hierarchies
related to the polynomial Lie algebra $gl_\infty^{(n)}$.} J.Geom.Phys. 57(2007)435--447.

\bibitem{M1998} G.Misiolek, \emph{A shallow water equation as
a geodesic flow on the Bott-Virasoro group}, J.Geom. Phys., 24 (1998)203--208.


\bibitem{OK1987} V.Yu.Ovsienko and B. Khesin, \emph{The (super) KdV  equation as an Euler equation},
 Funct.Anal. Appl., 21 (1987) 329--331.

\bibitem{PS1986} A.Pressley and G.Segal, \emph{Loop Groups}, Oxford University Press, Oxford, 1986.

\bibitem{Qu2013}C.Z.Qu, J.F.Song and R.X.Yao, \emph{Multi-Component Integrable Systems
and Invariant Curve Flows in Certain Geometries}, Symmetry Integrability Geom. Methods Appl.
 (SIGMA), 9 (2013) Paper 001, 19 pp.

  \bibitem{IS2013} Ian A.B.Strachan and B.M.Szablikowski,
\emph{Novikov algebras and a classification of multicomponent Camassa-Holm equations},
 To appear in Stud.in Appl.Math. (arXiv:1309.3188.)

 \bibitem{SZ2014} Ian A.B.Strachan and D.Zuo,
\emph{Frobenius manifolds and Frobenius algebra-valued Integrable systems},  arXiv:1403.?.

\bibitem{Zuo-2010-1}D.Zuo, \emph{A two-component $\mu$-Hunter-Saxton equation}, Inverse Problems, 26
 (2010) 085003, 9pp.

 \bibitem{Zuo-2013}D.Zuo, \emph{Euler Equations Related to the Generalized Neveu-Schwarz Algebra},
 Symmetry Integrability Geom. Methods Appl. (SIGMA), 9 (2013) Paper 045, 12pp.

 \bibitem{Zuo2013} D.Zuo, \emph{Local matrix generalizations of W-algebras}, arXiv:1401.2216.




\end{thebibliography}
\end{document}